\documentclass[11pt]{article}
\usepackage{ifthen}
\usepackage{xcolor}
\usepackage{fullpage}
\usepackage{pstricks}
\usepackage{pst-plot,pst-node}
\usepackage{float}
\usepackage{amsthm,amsmath,amssymb}
\usepackage{verbatim}

\newtheorem{lemma}{Lemma}
\newtheorem{definition}{Definition}
\newtheorem{mechanism}[definition]{Definition}
\newtheorem{theorem}{Theorem}
\newtheorem{corollary}{Corollary}

\newcommand{\floor}[1]{\lfloor#1\rfloor}

\newcommand{\ceil}[1]{\lceil#1\rceil}

\newcommand{\R}[1]{\mathrm{#1}}
\newcommand{\B}[1]{\mathbf{#1}}

\newcommand{\figcoordnorm}{\psaxes[labels=none,ticks=x]{->}(14,6)
\psline[linestyle=dotted,dotsep=2pt](1,0)(1,6)
\psline[linestyle=dotted,dotsep=2pt](2,0)(2,6)
\psline[linestyle=dotted,dotsep=2pt](3,0)(3,6)
\psline[linestyle=dotted,dotsep=2pt](4,0)(4,6)
\psline[linestyle=dotted,dotsep=2pt](5,0)(5,6)
\psline[linestyle=dotted,dotsep=2pt](6,0)(6,6)
\psline[linestyle=dotted,dotsep=2pt](7,0)(7,6)
\psline[linestyle=dotted,dotsep=2pt](8,0)(8,6)
\psline[linestyle=dotted,dotsep=2pt](9,0)(9,6)
\psline[linestyle=dotted,dotsep=2pt](10,0)(10,6)
\psline[linestyle=dotted,dotsep=2pt](11,0)(11,6)
\psline[linestyle=dotted,dotsep=2pt](12,0)(12,6)
\psline[linestyle=dotted,dotsep=2pt](13,0)(13,6)
\psline[linestyle=dotted,dotsep=2pt](0,3.2)(3,3.2)\rput(-0.5,3.2){$iv_i$}
\psline(-0.1,3.2)(0.1,3.2)
\rput(0,6.5){revenue}\rput(14,-0.5){index}
}

\newcommand{\figcoord}{\psaxes[labels=none,ticks=x]{->}(14,6)
\psline[linestyle=dotted,dotsep=2pt]{-}(1,-0.5)(1,6)\psline[linestyle=dotted,dotsep=2pt]{-}(2,-0.5)(2,6)\psline[linestyle=dotted,dotsep=2pt]{-}(3,-0.5)(3,6.5)\psline[linestyle=dotted,dotsep=2pt]{-}(5,-0.5)(5,6)\psline[linestyle=dotted,dotsep=2pt]{-}(7,-0.5)(7,6.5)\psline[linestyle=dotted,dotsep=2pt]{-}(10,-0.5)(10,6)
\psline[linestyle=dotted,dotsep=2pt]{-}(0,0)(12,0.93)
\psline[linestyle=dotted,dotsep=2pt]{-}(0,0)(12,1.35)
\psline[linestyle=dotted,dotsep=2pt]{-}(0,0)(12,1.94)
\psline[linestyle=dotted,dotsep=2pt]{-}(0,0)(12,2.79)
\psline[linestyle=dotted,dotsep=2pt]{-}(0,0)(12,4.02)
\psline[linestyle=dotted,dotsep=2pt]{-}(0,0)(12,5.79)
\psline[linestyle=dotted,dotsep=2pt]{-}(0,0)(8.64,6)
\psline[linestyle=dotted,dotsep=2pt]{-}(0,0)(6.00,6)
\rput(0,6.5){revenue}\rput(14,-0.5){index}
}
\newcommand{\figrevs}{\psline[linestyle=none,linewidth=2pt,showpoints=true,dotstyle=*](0,0)(1,2.4)(2,3.4)(3,3.2)(4,2.6)(5,3.25)(6,3.9)(7,4.2)(8,2.8)(9,3.1)(10,1.8)(11,1.4)(12,1)(13,0)\psline[linestyle=solid,linewidth=2pt](0,0)(1,2.4)(2,3.4)(7,4.2)(14,4.2)\rput(13.5,4.5){$\revs$}}
\newcommand{\figerevs}{
\psline[linestyle=dotted,linewidth=2pt](0,0)(3,3)(7,3.38)(14,3.38)\rput(13.5,3.08){$\erevs$}\psline[linestyle=none,linewidth=2pt,showpoints=true,dotstyle=square](0,0)(1,1)(2,2)(3,3)(4,1.93)(5,2.41)(6,2.90)(7,3.38)(8,0)(9,0)(10,0)(11,0)(12,0)(13,0)
\psline[linestyle=none,linewidth=2pt,showpoints=true,dotstyle=x](0,0)(3,3)(3,2.08)(7,3.38)(7,2.34)(7,1.63)(10,1.61)(10,1.13)(10,0.77)(10,0)}

\newcommand{\integers}{{\mathbb Z}}

\DeclareMathOperator{\CCEPE}{CCEPE}

\DeclareMathOperator{\CrossConsens}{CrossConsens}
\DeclareMathOperator{\Consens}{Consens}
\newcommand{\crossconsens}{\CrossConsens}
\newcommand{\consens}{\Consens}
\newcommand{\shrand}{\sigma}
\newcommand{\mech}{{\cal M}}

\newcommand{\stat}{s}
\newcommand{\cout}{\tilde{\stat}}
\newcommand{\sout}{I}

\newcommand{\super}[1]{^{(#1)}}

\newcommand{\rev}{R}
\newcommand{\revs}{{\mathbf \rev}}
\newcommand{\revi}[1][i]{{\rev_{#1}}}

\newcommand{\erev}{\tilde{R}}
\newcommand{\erevs}{{\mathbf \erev}}
\newcommand{\erevi}[1][i]{{\erev_{#1}}}

\newcommand{\val}{v}
\newcommand{\vals}{{\mathbf \val}}
\newcommand{\valsmi}[1][i]{{\mathbf \val}_{-#1}}
\newcommand{\vali}[1][i]{{\val_{#1}}}

\newcommand{\evirt}{\tilde{\virt}}

\newcommand{\evirts}{{\boldsymbol\evirt}}

\newcommand{\virt}{\phi}
\newcommand{\virti}[1][i]{\virt_{#1}}

\newcommand{\price}{p}
\newcommand{\prices}{{\mathbf \price}}
\newcommand{\pricei}[1][i]{{\price_{#1}}}

\newcommand{\alloc}{x}
\newcommand{\allocs}{{\mathbf \alloc}}

\newcommand{\alloci}[1][i]{\alloc_{#1}}

\DeclareMathOperator{\Roperator}{R}
\newcommand{\RC}[1][]{\Roperator\ifthenelse{\not\equal{}{#1}}{^{#1}}{}}
\newcommand{\IR}[1][]{\bar{\Roperator}\ifthenelse{\not\equal{}{#1}}{^{#1}}{}}

\DeclareMathOperator{\PFOoperator}{PFO}
\newcommand{\PFO}[1][]{\PFOoperator\ifthenelse{\not\equal{}{#1}}{^{#1}}{}}

\newcommand{\VV}[1][]{\Phi\ifthenelse{\not\equal{}{#1}}{^{#1}}{}}
\newcommand{\ivv}[1][]{\bar{\VV}\ifthenelse{\not\equal{}{#1}}{^{#1}}{}}

\DeclareMathOperator{\ICoperator}{IC}
\newcommand{\IC}[1]{\ICoperator\ifthenelse{\not\equal{}{{#1}}}{^{{#1}}}{}}

\DeclareMathOperator{\PEoperator}{PE}
\newcommand{\PE}[1]{\PEoperator_{#1}}

\DeclareMathOperator{\EFoperator}{EF}
\newcommand{\EF}[1]{\EFoperator\ifthenelse{\not\equal{}{#1}}{^{{#1}}}{}}

\DeclareMathOperator{\EFO}{EFO}

\DeclareMathOperator{\Vic}{Vic}

\newcommand{\eval}{\tilde{v}}
\newcommand{\evals}{{\mathbf \eval}}
\newcommand{\evali}[1][i]{{\eval_{#1}}}

\newcommand{\ealloc}{\tilde{x}}
\newcommand{\eallocs}{{\mathbf \ealloc}}
\newcommand{\ealloci}[1][i]{\ealloc_{#1}}

\begin{document}

\title{Mechanism Design via\\Consensus Estimates, Cross Checking, and Profit Extraction}
\author{
  Bach Q. Ha\\
  \small{Northwestern University}\\
  \texttt{\small{bach@u.northwestern.edu}}
\and
  Jason D. Hartline\\
  \small{Northwestern University}\\
  \texttt{\small{hartline@eecs.northwestern.edu}}
}
\date{}
\maketitle

\section{Introduction}

%
%
There is only one technique for {\em prior-free optimal mechanism design} that
generalizes beyond the structurally benevolent setting of digital goods.  This
technique uses random sampling to estimate the distribution of agent values and
then employs the {\em Bayesian optimal mechanism} for this estimated
distribution on the remaining players.  Though quite general, even for digital
goods, this random sampling auction has a complicated analysis and is known to
be suboptimal.
To overcome these issues we generalize the consensus technique from \cite{GH03}
to structurally rich environments that include, e.g., single-minded
combinatorial auctions.

%
%
The classical economic theory of mechanism design is Bayesian: it is assumed
that the preferences of the agents are drawn at random from a known probability
distribution and the designer aims to optimize their objective
in expectation over this randomization.  This leads to mechanisms that are
tailored to the distributional setting.  In contrast prior-free mechanism
design looks at mechanisms that perform well without knowledge or assumptions
on agent preferences.  While these mechanisms do not perform as well as ones
tailored to the distribution, in many setting they provide good approximations,
and are more robust.

%
%
The simplest environment in which to explore mechanism design is that of
selling a {\em digital good}, i.e., where the seller has no constraint over the
subsets of agents that can be served simultaneously.  Recent contributions to
the literature on prior-free mechanism design have focused on extending results
for digital good environments to ones that are more structurally rich.  From
least-general to most-general, these include {\em multi-unit environments},
where there are a given $k$ number of units available for sale (i.e., any
subset of the agents of size at most $k$ can be served); {\em matching
environments}, where feasible sets correspond to one side of a bipartite
matching; and {\em downward-closed environments}, where the only constraint on
feasible sets is that any subset of a feasible set is feasible.

%
%
The only prior-free mechanisms known to give good approximations for general
downward-closed environments are variants of the random sampling auction.  This
auction first gathers distributional information from a random sample of the
agents and then simulates the Bayesian optimal auction for the empirical
distribution on the remaining agents.  The agents in the sample are always
rejected.  Tight analysis of the random sampling auction is difficult,
upper-bounds on its approximation factor are 4.68, 25, 50, and 2560 for digital
good \cite{AMS09}, multi-unit~\cite{DH09}, matching~\cite{HY11}, and downward
closed environments~\cite{HY11}, respectively.  Other mechanism design
techniques give 3.25 and 6.5-approximations for digital-good~\cite{HM05} and
multi-unit environments~\cite{HY11}, respectively, and notably the 3.25
approximation for digital goods surpasses the lower-bound of 4 which is known
for the random sampling auction.  These limitations suggest the need to
consider other techniques for obtaining good approximations for general
downward-closed environments.

%
%
To obtain good approximation mechanisms for general downward-closed
environments, we generalize the digital good auction technique of consensus
estimates from~\cite{GH03}.  The two main ingredients of this approach are a
{\em profit extraction} mechanism and a {\em consensus} function.  Given a
target profit, the profit extraction mechanism should approximate the target
if the target is less than the optimal revenue possible.  If we had a good
estimate of the revenue, we could then obtain a good revenue with the profit
extractor.  The consensus function is used to get an estimate of the revenue
from the reports of the agents in a way that is non-manipulable.  In
particular, for each agent we can calculate the optimal profit from the other
agents, plug this profit into the consensus function, and with high
probability the estimated profit produced will be the same for all agents.  We
can then simulate the profit extraction mechanism for each agent with their
consensus estimate.  If the estimates agree, the result of this simulation is
the agreed-upon profit, otherwise, it is at least zero.

%
%
There are two main challenges to extending this approach for general
downward-closed environments.  The first challenge is in designing a
profit-extraction mechanism for these environments.  Our profit
extraction mechanism will be parameterized by a {\em revenue curve},
the revenue as a function of number of winners (without taking into
account any feasibility constraints).  Given a target revenue curve
that is below the actual revenue curve, our mechanism obtains revenue
comparable to that which would be obtained by the optimal mechanism on
the input that corresponds to the target revenue curve.  The second
challenge is in ensuring infeasible outcomes are not produced in the
case that the mechanism does not have a consensus.  Note that for
digital good environments, there is no feasibility constraint that
could be violated when the estimates do not reach consensus.  The same
is not so for general downward-closed environments.  A parameterized
mechanism (such as a profit extractor) is of course required to always
produced feasible outcomes.  However, if we determine the outcome for
each agent by simulating the parameterized mechanism with different
parameters, the combined outcome may not be feasible.  To address this
potential inconsistency we give a {\em cross checking} approach for
identifying a subset of agents which which consensus is achieved.

%
%
Our mechanism is a $30.4$-approximation in general downward-closed
environments.

\paragraph{Related Work.}
This paper derives its framework for prior-free mechanism design and
analysis from Hartline and Yan~\cite{HY11}.  It extends the
profit-extraction and consensus techniques from Goldberg and
Hartline~\cite{GH03,GH05}.  Other than these, the closest related work
to ours is that of Dhangwatnotai, Roughgarden, and Yan~\cite{DRY10}.
They consider abstract service provision in downward-closed
environments with the added assumption that the values of the agents
are distributed according to a unknown distribution that satisfies a
standard {\em monotone hazard rate} assumption.  Under this
assumption, they give an (essentially) $4$-approximation mechanism.
In contrast, our mechanism gives a worse bound, but does so without
the distributional assumption.

\paragraph{Organization.}
In Section~\ref{sec:prelim} we will formally describe our auction
environment, design and analysis framework, and review the consensus
technique.  In Section~\ref{sec:cross-check} we describe our
cross-checking approach as it applies to obtaining a consistent
consensus estimate.  In Section~\ref{sec:PE} we describe a mechanism
for extracting the profit suggested by a given target revenue curve.
In Section~\ref{sec:consensus} we describe an approach for obtaining a
consensus estimate on revenue curves.  Finally, in
Section~\ref{sec:mechanisms} we combine the three parts to give a good
mechanism and analyze its performance.


%
%

%
%

%
%


%
%


%
%

%
%

%
%

%
%



\section{Preliminaries}
\label{sec:prelim}

Here we describe the abstract setting in which we consider mechanism
design and the structural tools that we will be using to design and
analyze mechanisms.

\paragraph{Model.}

Let $N$ be a set of $n\geq 2$ bidders. Each bidder $i\in{N}$ has a
private \emph{valuation} $v_i$ for receiving some abstract service.  A
bidder $i$, upon reporting his valuation, will be served with a
probability $x_i$ and charged a payment of $p_i$.  We denote the {\em
  valuation profile}, {\em allocation vector}, and {\em payment
  vector} by $\B{v}=(v_1,\ldots,v_n)$,
$\B{x}=\{x_1,\ldots,x_n\}\in[0,1]^n$, and $\B{p}=(p_1,\ldots,p_n)$
respectively.  Without loss of generality we index the agents in
decreasing order of value, i.e., $\vali \geq \vali[i+1]$.

There is a {\em feasibility} constraint which describes the subsets of the agents that can be served simultaneously.  We assume that this
feasibility constraint is downward closed, i.e., any subset of a
feasible set is feasible.  As described in the introduction, many
common environments for mechanism design are downward-closed.  

We allow the feasibility constraint to be probabilistic, i.e., given
by a convex combination of downward closed set systems.  We assume the
following semantics for a randomized feasibility constraint: (i)
agents bid, (ii) the randomization over the feasibility constraint is
realized, and then (iii) the mechanism runs on the reported valuations
and the realized set system.  For the purpose of calculating revenue
the allocation $\B{x}$ and payments $\B{p}$ are taken in expectation
over the randomization in the mechanism and the set system.  That
said, the mechanisms we consider will be incentive compatible even
when the order of steps (i) and (ii) are reversed.

A fundamental and potentially restrictive assumption that we will make
is one of symmetry.  Given an asymmetric set system we can always make
it symmetric by randomly permuting the identities of the agents.  This
assumption is akin to standard assumptions for the secretary problem
and in settings where one might consider the agents to be a priori
identical (e.g., if there values were drawn from an
i.i.d.~distribution) it is without loss.  The resulting feasibility
constraint we refer to as a {\em downward-closed permutation
  environment}.

Our mechanisms will be based on simple algorithms.  Given weights
$\B{w} = (w_1,\ldots,w_n)$ indexed in non-increasing order, we assume
we have an algorithm for selecting a feasible set to optimize the sum
of the selected weights (for the realized set system) with ties broken
randomly.  As suggested above, $x_i$ will denote the probability (over
randomization in the set system and random tie-breaking) that
the $i$th largest weight is selected.  Clearly $\B{x}$ maximizes $\sum_i w_i
x_i$ subject to feasibility and so we will refer to $\B{x}$ as the
{\em maximizer} for weights $\B{w}$.

We assume the standard {\em risk-neutral} {\em quasi-linear} utility
model, i.e., an agent $i$ wishes to maximize their expected utility
which is given by $u_i=v_ix_i-p_i$. We will focus solely on incentive
compatible (IC) mechanisms; which means for any agent, reporting their
true valuation would be a dominant strategy; and we assume that agents
follow this dominant strategy.  We view a mechanism as a function from
reports to allocation and payments and denote these functions by
$\allocs(\cdot)$ and $\prices(\cdot)$.  A mechanism is incentive
compatible if and only if \cite{M81}:
\begin{enumerate}
\item $\alloci(\vals)$ is monotone non-decreasing in $\vali$,
  and
\item $\pricei(\vals)$ satisfies the payment identity:
\begin{align}
\label{eq:icpi}
\pricei(\vals) = \vali \alloci(\vals) -
  \int_0^{\vali}\alloci(z,\valsmi)\,dz,
\end{align}
\end{enumerate}
where $(z,\valsmi)$ is the valuation profile with $\vali$ replaced with $z$. 
The second equation is referred to as the {\em payment identity}.
When we give a mechanism we will describe only the {\em allocation
  rule} and infer the {\em payment rule} from this
identity.\footnote{Of course such a payment can be easily calculated,
  e.g., with techniques from Archer et al.~\cite{APTT03}.}  We will
denote the expected payment for agent $i$ with allocation rule
$\allocs(\cdot)$ as $\IC{\allocs}_i(\vals)$ and the total revenue as
$\IC{\allocs}(\vals)$.

\paragraph{Revenue Analysis.}
We adopt the framework from Hartline and Yan~\cite{HY11} wherein the
revenue of the designed mechanism is compared to the {\em envy-free
  revenue} benchmark.  We will denote the maximum envy-free revenue by
$\EFO(\vals)$.  For technical reasons we define our benchmark to be
$\EFO \super 2 (\vals) =
\EFO(\vali[2],\vali[2],\vali[3],\ldots,\vali[n])$. The goal of such a
design and analysis framework is then to give a mechanism that obtains
a revenue that is a good approximation to the benchmark $\EFO \super 2
(\vals)$ in worst case over valuation profiles $\vals$.

Envy-free revenue is defined only for allocation vectors $\allocs$
that are monotone, i.e., $\alloci \geq \alloci[i+1]$.  The envy-free
payment for agent $i$ with monotone allocation $\allocs$ on $\vals$ is
given by~\cite{HY11}:
\begin{align}
\label{eq:efpi}
\EF{\allocs}_i(\vals) 
   &= \sum\nolimits_{j \geq i}^n \vali[j] \cdot (\alloci[j] - \alloci[j+1]) 
\end{align}
The envy-free revenue, denoted $\EF{\allocs}(\vals)$, is the sum of
the envy-free payments.

The envy-free revenue can be understood structurally in terms of the
{\em revenue curve} and {\em virtual values}.  The revenue curve
$\revs$ for $\vals$ describes the optimal revenue as a function of the
number of agents served (when feasibility constraints are ignored).
The $i$th coordinate of the revenue curve, $\revi$, can be calculated
by evaluating at $i$ the smallest concave non-decreasing function that
contains the point set $\{(i,i\vali)\,:i \in \{1,\ldots,n\}\}$.  The
virtual value at $i$ is the left-slope of this function, i.e., $\virti
= \revi - \revi[i-1]$.  See Figure~\ref{fig:prelim}.


\begin{lemma}[Hartline and Yan \cite{HY11}]
\label{lemma:EF}
The envy-free revenue of monotone allocation $\allocs$ satisfies
\begin{align}
\EF{\allocs}(\vals)
&=\sum\nolimits_{i=1}^n\virti\cdot \alloci
=\sum\nolimits_{i=1}^n\revi\cdot(\alloci-\alloci[i+1])\label{eqn:EF}
\end{align}
as long as $\alloci=\alloci[i+1]$ whenever $\virti = \virti[i+1]$.
\end{lemma}
The optimal envy-free revenue, $\EFO(\vals)$, can be found from
Lemma~\ref{lemma:EF}, in particular, by optimizing {\em virtual
  surplus}, i.e., $\sum_i \virti \alloci$, with random tie-breaking.
Random tie-breaking results in an allocation $\allocs$ that satisfies
$\alloci=\alloci[i+1]$ whenever $\virti = \virti[i+1]$.

Notice that for the same allocation rule $\allocs$, the envy-free
payments \eqref{eq:efpi} and incentive compatible payments
\eqref{eq:icpi} are distinct, i.e., $\IC{\allocs}_i(\vals) \neq
\EF{\allocs}_i(\vals)$.

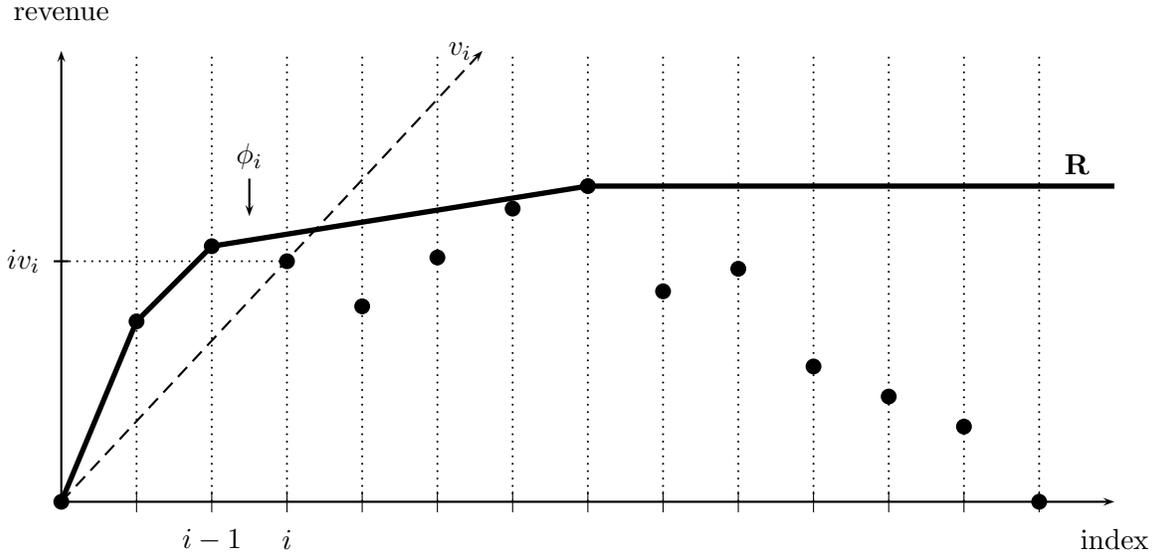
\begin{figure}[ht]
\begin{center}
\begin{pspicture}(0,-0.5)(14,6)
\figcoordnorm
\figrevs
\rput(2,-0.5){$i-1$}\rput(3,-0.5){$i$}
\rput(2.5,4.6){$\virti$}\psline{->}(2.5,4.3)(2.5,3.8)
\psline[linestyle=dashed]{->}(0,0)(5.6,6)\rput(5.3,6){$v_i$}
\end{pspicture}
\caption{Agents are indexed in decreasing order by value.  The point
  set $\{(i,i\vali) \,:\, i \in \{1,\ldots,n\}\}$ is depicted.  The
  revenue curve is the smallest non-decreasing concave function that
  upper-bounds this point set.  The value of agent $i$ can be
  represented by a line from the origin with slope $\vali$; the
  virtual value of agent $i$ is the left-slope of the revenue curve at
  $i$.  }\label{fig:prelim}
\end{center}
\end{figure}

\paragraph{Consensus estimates.} 

A central ingredient in our approach is the technique of {\em
  consensus estimates} that was introduced by Goldberg and Hartline
\cite{GH03}. A {\em consensus function} maps shared
randomness and a statistic to an estimate of the statistic.  The
objective of such a consensus function is that, when applied individually to each of a set
of statistics that are within some bounded range, with high
probability (in the shared randomness) the estimates will coincide,
i.e., there will be a {\em consensus} among the {\em estimates}.

\begin{definition}
\label{d:consensus}
For implicit parameter $c > 1$ and shared randomness $\shrand \sim U[0,1]$, the {\em consensus function} on statistic $\stat$ is,
$$
\consens(\shrand,\stat) = \floor{\stat}_{\{c^{\shrand+d} \,:\,d \in \integers\}},
$$ 
where $\floor{\stat}_S$ denotes $\stat$ rounded down to the nearest element
of $S$.
\end{definition}

\begin{lemma} \label{lemma:prob} \cite{GH03} 
For $c \geq \beta$, The probability (over randomization of $\shrand$)
that the consensus function is constant on interval $[\stat/\beta,\stat]$ is
$1-\log_c \beta$.
\end{lemma}

\section{Cross-checking}

\label{sec:cross-check}

Given some statistic on valuation profiles, we will be using the
consensus function (Definition~\ref{d:consensus}) to get an estimate of
this statistic, e.g., by calculating
$\consens(\shrand,\stat(\valsmi))$ for each $i$ where $\stat(\cdot)$
calculates the statistic for a valuation profile.  Notice that if we
had some mechanism $\mech_{\stat}$ that was parameterized by statistic
$\stat$ then, if there is consensus, simulations of
$\mech_{\consens(\shrand,\stat(\valsmi))}$ to determine the
  allocation $\alloci$ and payment $\pricei$ for agent $i$ are
  internally consistent.  I.e., the outcome produced by
  $\mech_{\stat}$ is feasible for any $\stat$; therefore, so is the combined
  outcome.  Unfortunately, when consensus is not achieved then these
  simulations may not be consistent.  

In this section we give a method of cross-checking to ensure that
consistent estimates of the statistic for some subset of the agents.
For environments with downward closed feasibility constraints, such a
method can be used in mechanism design as agents outside this
consistent subset can be rejected.

\begin{definition} 
\label{d:cross-check}
For shared randomness $\shrand$, statistic $\stat$, and consensus function $\consens$, and valuation profile $\vals$ calculate the following:
\begin{enumerate}
\item For all pairs $i \neq j \in \{1,\ldots,n\}$, calculate the
  consensus estimate $\cout_{i,j} = \consens(\shrand,\stat(\valsmi[i,j]))$.
\item $\sout$ is the set of agents $i$ that have consensus on
  $\cout_{i,j}$ for all $j$; or $\emptyset$ if no such $i$ exists.
\item $\cout$ is the consensus $\cout_{i,j}$ of any $i \in I$ and any
  $j$ (they are all the same).
\end{enumerate} 
The {\em cross-checked consensus function} is defined as
$$\crossconsens(\shrand,\stat,\vals) = (\cout,\sout).$$
\end{definition}

Cross-checked consensus estimates are non-manipulable in a strong
sense.  Whether or not an agent $i$ is in $\sout$ is not a function of
that agents value.  Furthermore, the final consensus is not a function
of the report of any agent $i \in \sout$.  This implies that mechanisms
which take the following form are incentive compatible.

\begin{definition}[cross-checked consensus estimate composition] 
Given an incentive compatible mechanism $\mech_{\stat}$ that is
parameterized by some statistic $\stat$ and a consensus function
$\consens$ for the statistic, compose them as follows:
\begin{enumerate}
\item Calculate cross-checked consensus estimate $(\cout,\sout) =
  \crossconsens(\shrand,\stat,\vals)$.
\item Simulate incentive compatible mechanism $\mech_{\cout}$ on $\vals$.
\item For agents $i \in \sout$ output result of simulation, reject all others.
\end{enumerate}
\end{definition}

\begin{theorem} 
Mechanisms produced by the cross-checked consensus estimate
construction are incentive compatible.
\end{theorem}




\section{Profit Extraction Mechanism}
\label{sec:PE}

Using the techniques in the previous section we will construct a
consensus estimate for the revenue curve.  In this section we will
show how to design a mechanism with good revenue that is parameterized
by an approximation of the revenue curve.  Such a mechanism is termed
a {\em profit extractor}. 
Given a
target revenue curve that is upper-bounded by our actual revenue
curve, this mechanism will obtain at least the optimal envy-free
revenue for the target revenue curve.  The target revenue curve will
be provided to the mechanism in the form of the valuation profile
$\evals$ that generates it.  We will denote by $\erevs$ and $\evirts$
the revenue curve and virtual values for $\evals$.

\begin{mechanism}[Profit Extractor, $\PE{\evals}$] \label{mech:PE}
  Parameterized by non-increasing valuation vector $\evals$:
  \begin{enumerate}
    \item Sort the bids in a non-increasing order, break ties arbitrary.
    If $\evali>\vali$ for some $i$, reject everyone and charge nothing.
    \item Assign weights $\evirts$ to agents in the same order as their values.
    \item Serve the set of agents to maximize the sum of their
      assigned weights.
  \end{enumerate}
\end{mechanism}

We will show that the IC revenue obtained from the profit extractor
for $\evals$ on $\vals$ is higher than the optimal envy-free revenue
for $\evals$.  Furthermore, for appropriately chosen $\evals$, this
revenue approximates the optimal envy-free revenue for $\vals$.

\begin{lemma}\label{lemma:PE:ICvsEFO}
  For any $\evals\leq\vals$, the revenue of the profit extractor on
  $\vals$ is at least the envy-free optimal revenue for $\evals$.
  Moreover, the inequality holds on each agent's payment, i.e.,
  $\IC{\PE{\evals}}_i(\vals)\geq\EFO_i(\evals)$.
\end{lemma}

\begin{proof}
We show the second condition of the lemma for any agent $i$ (the first
condition follows).  Let $\eallocs$ be the allocation for
$\EFO(\evals)$.  This is the same allocation as used by $\PE{\evals}$
unless $\vals \geq \evals$ fails to hold.

First, notice that the IC payments, from equation \eqref{eq:icpi}, and
EF payments, from equation \eqref{eq:efpi}, correspond to the area in
the region bounded by $\alloc \leq \ealloci$ (above), $\val \geq 0$
(left), and the ``allocation rule'' (bottom right).  For IC payments,
this allocation rule is the probability that the agent is served for
any possible misreport $z$.  For EF payments, this ``allocation rule''
is the smallest monotone function that upper-bounds the point set
$\{(\evali, \ealloci) \,:\, i \in \{1,\ldots,n\}\}$.  To prove the
lemma we need only show that the IC allocation rule gives a weaker
bound than the EF ``allocation rule.''

For any $j > i$, the EF allocation rule drops from $\ealloci[j]$ to
$\ealloci[j+1]$ at $\evali[j]$.  We claim that the IC allocation rule
makes the same drop but at a value that is at least $\evali[j]$.  To
see this claim, consider the minimum bid $z$ that agent $i$ can make
to secure allocation probability at least $\ealloci[j]$.  First, she
must out bid agent $j+1$, i.e., $z \geq \vali[j+1]$.  Assuming she
out bids $j+1$, she must also $z \geq \evali[j]$ otherwise, the
condition $(z,\vals_{-i}) \geq \evals$ is not met and all agents are
rejected.  Therefore, if $i$ bids $z < \evali[j]$ then she is
allocated with probability at most $\ealloci[j+1]$.  Hence the IC
allocation rule drops from $\ealloci[j]$ to (at most) $\ealloci[j+1]$
at value at least $\evali[j]$.
\end{proof}

\begin{lemma}\label{lemma:PE:EFOvsEFO}
  For any $\evals$ and $\vals$ with $\erevs \geq \tfrac{1}{\beta} \revs$, the envy-free optimal revenue for $\evals$ is a $\beta$-approximation to that from $\vals$, i.e.,
$\EFO(\evals)\geq\tfrac{1}{\beta}\cdot\EFO(\vals).$
\end{lemma}

\begin{proof}
Let $\allocs$ and $\eallocs$ be the allocations $\EFO(\vals)$ and
$\EFO(\evals)$, respectively.
  \begin{align*}
    \EFO(\evals) 
    &= \sum\nolimits_i \erevi \cdot(\ealloci - \ealloci[i+1])
    \geq \sum\nolimits_{i}\erevi\cdot (\alloci-\alloci[i+1])
    \geq \tfrac{1}{\beta} \sum\nolimits_{i}\revi\cdot(\alloci-\alloci[i+1])
    = \tfrac{1}{\beta} \EFO(\vals).
  \end{align*}
 The first inequality follows from the optimality of $\eallocs$ for $\evals$ and Lemma~\ref{lemma:EF}.
 The second inequality follows from monotonicity of $\allocs$ and the
 assumption that $\forall i$, $\erevi\geq\tfrac{1}{\beta}\cdot\revi$.
\end{proof}

Combining these lemmas, we see that with the right $\evals$,
$\PE{\evals}(\vals)$ can approximate the optimal envy free revenue on
$\vals$.

\begin{theorem} \label{thm:PE}
  For any $\evals \leq \vals$ with $\erevs \geq \tfrac{1}{\beta}
  \revs$, the profit extractor for $\evals$ on $\vals$ is a
  $\beta$-approximation to the optimal envy-free revenue for $\vals$, i.e., $
  \IC{\PE{\evals}}(\vals) \geq \tfrac{1}{\beta}\cdot\EFO(\vals).  $
\end{theorem}

\makeatletter
\makeatother

\newcommand{\point}[2]{\mathrm{Q}_{#2}^{#1}}
\newcommand{\naj}[1][j]{n_{#1}}          
\newcommand{\enaj}[1][j]{\tilde{n}_{#1}} 
\newcommand{\pointj}[1][j]{Q_{#1}} 
\newcommand{\na}[2]{{n_{#2}(#1)}} 
\newcommand{\ena}[2]{{\tilde{n}_{#2}(\shrand,#1)}}
\newcommand{\crevs}[1]{\erevs(\shrand,#1)}
\newcommand{\crevi}[2]{\erevi[#2](\shrand,#1)}
\newcommand{\cej}[1][j]{{\cal E}_{#1}}   

\section{Consensus Estimates of Revenue Curves}
\label{sec:consensus}

Our objective now is to get a consensus estimate of the revenue curve.
We will express the estimate revenue curve, $\erevs$, in terms of an
estimate valuation profile, $\evals$, that generates it.

In addition to the implicit parameter $c>1$ in the definition of
consensus (Definition~\ref{d:consensus}) we will also use implicit
parameter $\alpha>1$ and a minimum required support $m \in
\integers_+$.  A statistic we will be interested in getting consensus
on is the number of agents with values at least $\alpha^j$ for any
given $j$.  We will use $\na{\vals}{j}$ to denote this statistic.  As
per our notation in the previous section, we will denote
$\ena{\vals}{j}=\ceil{\consens(\shrand,\na{\vals}{j})}$ where we round
the estimate up to the nearest integer because it is an integer
statistic.  Estimates that do not have the minimum required support of
$\enaj(\shrand,\vals) \geq m$ will be discarded.  We will use the
remaining estimates to construct an estimate of the
valuation profile $\evals(\shrand,\vals)$ and revenue curve
$\erevs(\shrand,\vals)$ as follows.

\begin{definition}[estimated revenue curve and valuation profile]
  \label{d:evals}
For any $j$ for which the estimates $\enaj$ of the number of agents
with values $\alpha^j$ is at least the minimum required support $m$,
define point $\pointj = (\enaj, \alpha^j\enaj)$.  The {\em estimated
  revenue curve}, $\erevs$, is the minimum non-decreasing concave
function that upper-bounds the point set $\{\pointj
\}_{j\in\integers}$ and the origin.  Let $j_k$ denote the
$k$\textsuperscript{th} largest index such that point $\pointj[j_k]$
is on $\erevs$.  The {\em estimated valuation profile}, $\evals$, has
$\enaj[j_{k}]-\enaj[j_{k-1}]$ values equal to $\alpha^{j_k}$.  Pad the
remainder of $\evals$ with zeros to get an $n$-vector.
\end{definition}

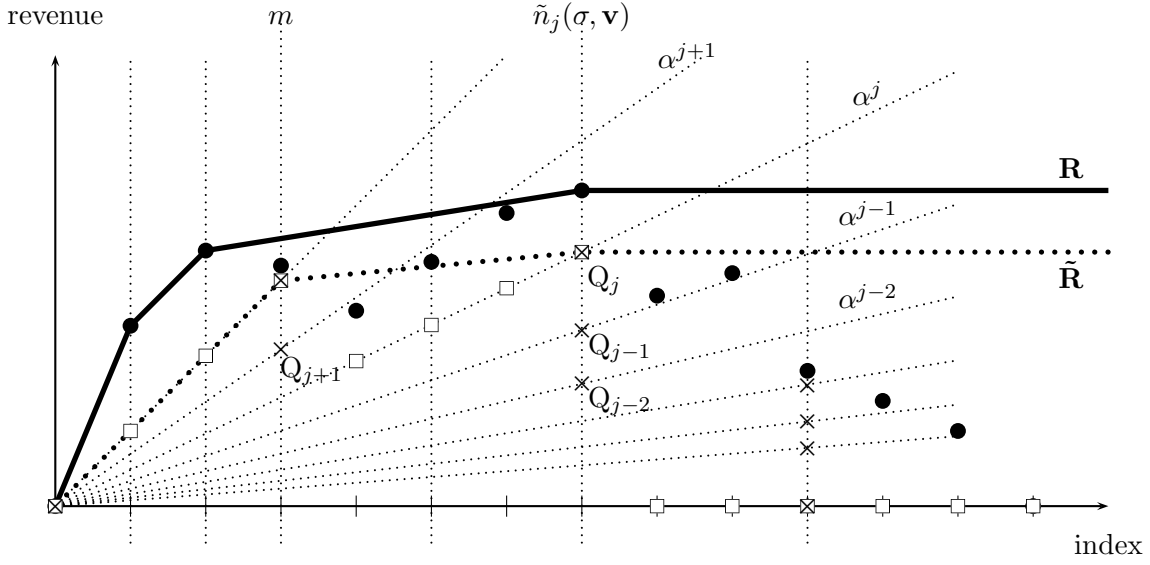
\begin{figure}[ht]
\begin{center}

\begin{pspicture}(0,-0.5)(14,6.5)
  \figcoord
  \figrevs
  \figerevs
  \rput(8.4,6){$\alpha^{j+1}$}\rput(3.4,1.8){$\R{Q}_{j+1}$}
  \rput(10.8,5.5){$\alpha^j$}\rput(7.3,3){$\R{Q}_j$}
  \rput(10.8,3.9){$\alpha^{j-1}$}\rput(7.5,2.1){$\R{Q}_{j-1}$}
  \rput(10.8,2.8){$\alpha^{j-2}$}\rput(7.5,1.4){$\R{Q}_{j-2}$}
  \rput(7,6.5){$\ena{\vals}{j}$}
  \rput(3,6.5){$m$}
\end{pspicture}

  \caption{The dotted vertical lines are $\ceil{c^{d+\shrand}}$ for
    $d\in\mathbb{Z}$; while the diagonal lines from the origin
    represent value $\alpha^j$ (as their slopes) for $j\in\mathbb{Z}$.
    Each value $\vali$ is represented with ``$\bullet$'' at point $(i,
    i\vali)$.  The line upper-bounding the black dots represents the
    original revenue curve, $\revs$.  For each diagonal line
    representing $\alpha^j$, $\naj(\vals)$ is the number of values
    that lie above it; the dotted vertical line immediately to the
    left of the right-most among these values represents
    $\ena{\vals}{j}$; finally the intersection between this line and
    $\alpha^j$ is the point $\pointj$, represented with ``$\times$''.
    $\erevs$ and $\evals$ are constructed from estimates
    $\enaj(\shrand,\vals) \geq m$ and their corresponding $\pointj$s.
    The thick dotted function upper-bounding these $\pointj$s is
    $\erevs$, and each $\evali$ is represented by the point
    $(i,i\evali)$ with a ``{\small$\Box$}''.}
  \label{fig:curvedefs}
\end{center}
\end{figure}

In the above construction $\evals$ is the smallest (point-wise)
valuation profile that has revenue curve $\erevs$, and furthermore,
$\evals\leq\vals$.  For statistical estimates $\enaj(\shrand,\vals)$
the estimated revenue curve and valuation profile will be denoted
$\evals(\shrand,\vals)$ and $\erevs(\shrand,\vals)$.  Our goal now is
to show that with high probability the estimated revenue curve (and
valuation profile) has consensus when a few agents, $S$, are omitted,
i.e., $\erevs(\shrand,\vals) = \erevs(\shrand,\vals_{-S})$.  To do
this we define a notion of relevance for statistics $j$ and show that
with high probability there is simultaneous consensus for all relevant
statistics.

\begin{definition}[$t$-consensus on $\vals$]\label{d:tcon}
  Given a fixed valuation vector $\vals$, a positive integer constant
  $t$ and a fixed choice of $\shrand$, the $j$th statistic has a \emph{$t$-consensus on $\vals$} if for every set $S\subset\{1,2,\ldots,n\}$ of
  no more than $t$ elements,
  \begin{equation*}
    \ena{\vals}{j} = \ena{\valsmi[S]}{j}.
  \end{equation*}
\end{definition}

\begin{definition}
\label{def:relevancy}
  For a given valuation vector $\vals$ and a positive integer $t$, the
  $j$th statistic $\naj$ is \emph{relevant} if there exists
  $\shrand\in[0,1]$, and a set $S\subseteq\{1,2,\ldots,n\}$ of no more
  than $t$ elements such that the point $\pointj(\shrand,\valsmi[S])$ is on
  $\erevs(\shrand,\valsmi[S])$.
\end{definition}

Notice that when $t$-consensus happens for a relevant statistic $j$, then
points $\pointj(\shrand,\vals)$ and $\pointj(\shrand,\vals_{-S})$ in
the construction of the estimated revenue curve are identical.

We now argue that the probability that any statistic $j$ has does not
have consensus is roughly proportional to $1/\naj(\vals)$, that for
relevant statistics the $\naj(\vals)$ values are geometrically
increasing, and thus the union bound implies that all estimates of
relevant statistics, and thus the estimated revenue curve, have
consensus with high probability.  This approach is adapted from
Goldberg and Hartline~\cite{GH05}.

\begin{lemma}\label{lemma:singleprob} For any $\vals$ and $\shrand \sim U[0,1]$,
  the probability that the $j$th statistic has a $t$-consensus on $\vals$ is at least
$
    1+\log_c \left(1 - \tfrac{t}{\naj(\vals)}\right).
$
\end{lemma}

\begin{proof}
  Observe that for every set $S$ of no more than $t$ elements,
  $\na{\vals}{j}-t \leq \na{\valsmi[S]}{j} \leq \na{\vals}{j}$.  These
  inequalities hold since when some bids are removed, the number of
  bids above any $\alpha^j$ decrease, but only by at most the size of
  the removed set.  Thus the probability that
  $\consens(\shrand,\na{\vals}{j})=\consens(\shrand,\na{\valsmi[S]}{j})$
  is at least $1-\log_c\tfrac{\na{\vals}{j}}{\na{\vals}{j}-t}$ as suggested by
  Lemma~\ref{lemma:prob}. The lemma follows from the power rule for
  logarithm.
\end{proof}

\begin{lemma}
\label{lemma:alpha}
  The values above successive relevant statistics are bounded by a
  geometrically increasing function: for any relevant statistic $j$,
  $\naj(\vals) \geq m\alpha^{r-j}$ where $r$ is the largest index of
  any relevant statistic.
\end{lemma}

\begin{proof}
   First note that the largest index of a relevant statistic $r$ is
   well defined.  For any $j$ that is relevant, it must be that
   $\naj(\vals) \geq m$; otherwise, $j$ would be discarded by the
   estimated revenue curve construction.  Thus, the largest index that
   may not be discarded is $\floor{\vali[m]}_{\{\alpha_j \,:\,j \in
     \integers\}}$.

  From the definition of $\ena{\valsmi[S]}{j}$, we have
  $\alpha^j\ena{\valsmi[S]}{j} \leq \alpha^j\na{\valsmi[S]}{j} \leq
  \alpha^j\na{\vals}{j}$.  Since $j$ is relevant, the corresponding
  point $\pointj(\shrand,\valsmi[S])$ must be higher than
  $\pointj[r](\shrand,\valsmi[S])$; therefore,
  $\alpha^j\ena{\valsmi[S]}{j}\geq\alpha^r\ena{\valsmi[S]}{r} \geq
  m\alpha^r$.  Combining this with the previous inequality, we have
  the desired claim.
\end{proof}

\begin{lemma}\label{lemma:totalprob}
  The probability of $t$-consensus at all relevant values is at least
  \begin{equation*}
    1+\log_c\left[1-\tfrac{t\alpha}{m(\alpha-1)}\right].
  \end{equation*}
\end{lemma}

\begin{proof}
  We will first bound the probability of consensus at one relevant value,
  then use the union bound to find the lower-bound of the probability of
  consensus at all relevant values.  For any relevant statistic $j$,
  let $\cej$ denote the event that $\naj$ has a $t$-consensus on $\vals$. 
  \begin{align*}
    \Pr[\cej]
    \geq 1+\log_c \left (1 - \tfrac{t}{\naj(\vals)}\right)
    \geq 1+\log_c\left(1-\tfrac{t}{m}\alpha^{-(r-j)}\right).
  \end{align*}
  The first inequality is from Lemma \ref{lemma:singleprob}, while the
  second inequality is from Lemma \ref{lemma:alpha}.  Let
  $R=\{j\,\:\,\text{$\naj$ is relevant}\}$.  The probability that all
  relevant statistics have $t$-consensus on $\vals$, using the union
  bound, is
  \begin{align*}
    \Pr[\text{$t$-consensus}]
    & \geq 1-\sum\nolimits_{j\in R}\Pr[\neg\cej]
      \geq 1+\sum\nolimits_{j\in R}\log_c\left(1-\tfrac{t}{m}\alpha^{-(r-j)}\right)\\
    & \geq 1+\sum\nolimits_{i\geq 0}\log_c\left(1-\tfrac{t}{m}\alpha^{-i}\right)
      =1+\log_c\left[
      \prod\nolimits_{i\geq 0}\left(1-\tfrac{t}{m}\alpha^{-i}\right)\right]\\
    & \geq 1+\log_c\left[1-\sum\nolimits_{i\geq 0}\tfrac{t}{m}\alpha^{-i}\right]
      = 1+\log_c\left[1-\tfrac{t}{m}\tfrac{\alpha}{\alpha-1}\right].
  \end{align*}
\end{proof}


The last thing that we need for our estimated revenue curves is for
them to be good estimates.  This follows directly from their
definition.

\begin{lemma}\label{lemma:ratio} For any $\shrand$ and $\vals$, 
  the consensus revenue curve $\erevs(\shrand,\vals)$ is a
  $c\alpha$-approximation of the revenue curve $\revs\super{m'}$ for
  $m' =
  \floor{mc}$ and truncated valuation profile $\vals\super{m'}=(\vali[m'],\ldots,
  \vali[m'],\vali[m'+1],\vali[m'+2],\ldots,\vali[n])$,
i.e., $\erevs(\shrand,\vals)\geq \tfrac{1}{c\alpha}
  \revs{\super{m'}}$.
\end{lemma}

\begin{proof}
It is sufficient to show that $\erevi \geq \tfrac{1}{c\alpha} i \vali$
for $i \geq m'$ as concavity of revenue curves would then imply the
lemma.  Consider then any index $i \geq m'$ and let $j$ be the index
of the statistic that satisfies $\alpha^j\leq\vali<\alpha^{j+1}$.
Since $\naj(\vals) \geq i$, by the definition of $\enaj$ and $m'$,
respectively, $\enaj \geq \ceil{i/c} \geq m$; therefore, statistic $j$
is not discarded in the first step of the construction of $\erevs$.
Furthermore, the point $\pointj = (\enaj,\alpha^j \enaj)$ is above and
to the left of $(i, \tfrac{1}{c\alpha} i \vali)$ because $\enaj \leq
i$, $\alpha^j \geq \vali/\alpha$, and $\enaj \geq i/c$.  Monotonicity
of $\erevs$, then, implies the desired $\erevi \geq \tfrac{1}{c\alpha}
i \vali$.
\end{proof}

\makeatletter
\let\na\@undefined
\let\ena\@undefined
\let\point\@undefined
\let\crevs\@undefined
\let\crevi\@undefined
\makeatother

\newcommand{\crevs}[1]{\erevs(\shrand,#1)}
\newcommand{\crevi}[2]{\erevi[#2](\shrand,#1)}

\section{Designed Mechanisms}

\label{sec:mechanisms}

We now proceed to define a mechanism that is a good approximation to
$\EFO \super 2 (\vals)$.  This mechanism will be a convex combination
of a primitive cross-checked consensus-estimate profit-extraction
mechanism and an extension of the Vickrey auction to downward-closed
permutation environments.

\begin{definition}\label{mech:ccepe'}
The {\em primitive cross-checked consensus-estimate profit-extraction}
mechanism, $\CCEPE'$ is the profit extraction mechanism $\PE{\evals}$
(Definition~\ref{mech:PE}) composed (Definition~\ref{d:cross-check})
with the valuation profile estimate (Definition~\ref{d:evals}).
$\CCEPE'$ is parameterized implicitly by $\alpha$, $c$, and $m$. 
\end{definition}

\begin{definition}\label{mech:vickrey}
The Pseudo-Vickrey auction, $\Vic$, serves the highest valued agent (and charge
her the second highest agent's value) if doing so is feasible with
respect to the set system, otherwise, serve no one.
\end{definition}

\begin{mechanism}\label{mech:ccepe}
The {\em cross-checked consensus-estimate profit-extraction}
mechanism, $\CCEPE$, is a convex combination of the Pseudo-Vickrey auction (with probability $p$) and $\CCEPE'$ (with probability $1-p$).  $\CCEPE'$ is parameterized implicitly by $\alpha$, $c$, $m$, and $p$. 

\end{mechanism}

The Pseudo-Vickrey mechanism is intended to obtain good revenue from
the highest-valued agents where as $\CCEPE'$ is intended to obtain
good revenue from the lower-valued agents.  The convex combination
obtains good revenue over all.  This analysis is given by the
following lemmas and theorem.

\begin{lemma}\label{l:ccepe'} For any $\vals$ and $m' = \floor{mc}$,
  the expected revenue of $\CCEPE'$ is a $\beta'$ approximation of approximates $\EFO\super{m'}$ with 
  \begin{align*}
    \beta' \leq 
    {c\alpha}
    \left[1+\log_c\left(1-\tfrac{2\alpha}{m(\alpha-1)}\right)\right]^{-1}.
  \end{align*}
\end{lemma}

\begin{proof}
  Let $(\evals,I)$ denote the outcome of the cross-checked consensus estimate
  of the valuation profile.  By Lemma~\ref{lemma:totalprob} (with
  $t=2$) all agents are cross-checked and $I = \{1,\ldots,n\}$ with
  probability at least $1+\log_c\left[1-\tfrac{2\alpha}{m(\alpha-1)}\right]$.
  In this case, $\evals \leq \vals$ (Definition~\ref{d:evals}) and
  $\erevs \geq \tfrac{1}{c\alpha} \revs \super {m'}$
  (Lemma~\ref{lemma:ratio}) therefore, the profit-extraction
  mechanism for $\evals$, $\PE{\evals}$, on $\vals$ obtains revenue at least
  $\frac{1}{c\alpha} \EFO \super {m'} (\vals)$ (Theorem~\ref{thm:PE}).
\end{proof}

\begin{lemma}\label{l:vic}
For any $\vals$, the Pseudo-Vickrey auction revenue is at least the
envy-free optimal payment of the highest-valued agent, i.e.,
$\IC{\Vic}(\vals) \geq \EFO \super 2 _ 1 (\vals)$.
\end{lemma}

\begin{proof}
Assume that $\vali[1] = \vali[2]$, this is without loss for this lemma
because both Pseudo-Vickrey's revenue and $\EFO \super 2$'s revenue is
the same on $\vals$ and $\vals \super 2 =
(\vali[2],\vali[2],\vali[3],\ldots,\vali[n])$.  The payment for agent
1 upon winning in Pseudo-Vickrey is $\vali[2] = \vali[1]$; the payment
upon winning in $\EFO \super 2$ is at most $\vali[1] = \vali[2]$.  The
probability that agent 1 wins in Pseudo-Vickrey is the highest of any
feasible allocation (because agent 1 wins whenever serving agent 1 is
feasible); in particular it is as high as that of $\EFO \super 2$.  Therefore,
the revenue from agent 1 in Pseudo-Vickrey is at least that of $\EFO \super 2$.
\end{proof}

\begin{theorem} For any $\vals$,
  $\CCEPE$ is a $\beta$-approximation to $\EFO\super{2}(\vals)$ where $\beta$ satisfies
  \begin{align*} 
   \beta \leq \max\left\{\tfrac{\floor{mc}}{p},\tfrac{c\alpha}{1-p}
    \left[1+\log_c\left(1-\tfrac{2\alpha}{m(\alpha-1)}\right)\right]^{-1}\right\}.
  \end{align*}
\end{theorem}

\begin{proof}
  We will separate our revenue into two parts; the first part is
  obtained from the top $m' = \floor{mc}$ agents, denoted $H =
  \{1,\ldots,m'\}$; and the second part is obtained from the remaining
  $n-m'$, denoted $L = \{m'+1,\ldots,n\}$.

  The contribution to the envy-free optimal revenue by the top agents
  satisfies $\EFO \super 2 _ H (\vals) \leq m' \IC{\Vic}(\vals)$.
  This bound follows from Lemma~\ref{l:vic} and the observation that
  envy-free payments are monotonically non-increasing in agent values.

  The contribution to the envy-free optimal revenue by the bottom
  agents satisfies $\EFO \super 2 _ L (\vals) \leq \EFO \super {m'}
  (\vals)$.  This bound follows as $\EFO \super {m'}$ could try to
  simulate the outcome of $\EFO \super 2$ and would then receive the
  same contribution to revenue from agents $L$ as $\EFO \super 2$;
  of course, its revenue from all agents must only be higher.

  In conclusion, $\EFO \super 2 (\vals) \leq m'\IC{\Vic}(\vals) + \EFO
  \super {m'} (\vals)$.

  The revenue of our mechanism, $\CCEPE$, the sum of a $\beta_1 =
  m'/p$ approximation to $m'\IC{\Vic}(\vals)$ and a $\beta_2 =
  \beta'/(1-p)$ approximation to $\EFO \super {m'} (\vals)$, with
  $\beta'$ as defined in Lemma~\ref{l:ccepe'}.  Therefore, it is a
  $\beta = \max(\beta_1,\beta_2)$ approximation to $\EFO \super
  2(\vals)$.
\end{proof}

We can optimize the parameters of $\CCEPE$ to obtain the following corollary.

\begin{corollary}
  For any $\vals$, $\CCEPE$ with
  $p=0.627$, $c=1.666$, $\alpha=2.734$ and $m=12$ is a
  $30.4$-approximation to $\EFO\super{2}(\vals)$.
\end{corollary}

\bibliographystyle{plain}
\bibliography{consensus}

\end{document}